\documentclass[12pt]{article}
\usepackage[utf8]{inputenc}
\usepackage{tikz}
\usepackage{graphicx}
\usepackage{url}
\usepackage{amsmath}
\usepackage{epstopdf}
\usepackage{amsthm}
\usepackage{amssymb}
\usepackage{mathrsfs}
\usepackage{natbib}
\usepackage[hmargin=2.3cm,vmargin=2.3cm,headsep=0.4cm]{geometry}
\usepackage[onehalfspacing]{setspace}
\usepackage{color}
\definecolor{darkblue}{rgb}{0,0,.5}
\usepackage[english]{babel}
\usepackage{booktabs}
\usepackage{array}
\usepackage{paralist}
\usepackage{subcaption}
\usepackage[format=hang,font=small,labelfont=bf]{caption}
\usepackage{fancyhdr}
\usepackage{sectsty}
\sectionfont{\Large}
\subsectionfont{\large}
\subsubsectionfont{\large}
\usepackage[colorlinks=true,linkcolor=black,urlcolor=darkblue,citecolor=darkblue]{hyperref}

\newtheorem{proposition}{Proposition}

\makeatletter
\renewenvironment{proof}[1][\proofname]
{\par\pushQED{\qed}\normalfont\topsep6\p@\@plus6\p@\relax\trivlist\item[\hskip\labelsep\bfseries#1\@addpunct{.}]\ignorespaces}
{\popQED\endtrivlist\@endpefalse}
\makeatother

\title{\textbf{The Optimal Reset-Hour of a Once-Daily Petrol Price Increase Limit}}

\author{Christoph Siemroth\footnote{Wivenhoe Park, Colchester, CO4 3SQ, UK. E-mail: christoph.siemroth@essex.ac.uk.} \\
Department of Economics \\
University of Essex}

\date{18 March 2026}

\begin{document}
\maketitle

\begin{abstract}
\noindent A German ministry recently proposed a limit of at most one price increase per day for petrol stations. At what time should the price reset be allowed in order to lower price levels the most throughout the day? To answer this question, I infer the share of price-sensitive consumers for every hour of the day from German petrol station price data, based on a simple spatial-competition model. I focus on weekdays, which are the relevant target because commuter demand is less flexible than weekend demand. Hourly petrol station prices peak at 07:00 and bottom out at 19:00. Given the inferred composition of price-sensitivity throughout the day and hourly passenger-car traffic frequencies as a proxy for quantity, I evaluate every possible reset-hour of the new policy. The lowest traffic-weighted average price is achieved by an 11:00 reset. With this reset-hour, the resulting equilibrium price throughout the day is constant. This would lead to lower prices in the morning but higher prices in the evening, harming price-sensitive consumers but benefiting morning commuters and firms. 

\vspace{0.3cm}

\vspace{0.1cm}
\end{abstract}
\newpage

\section{Introduction}
In March 2026, the Iran war pushed crude oil prices and petrol-station prices up throughout the world. Consumers in many countries called for mitigation in the form of price controls or petrol tax decreases. Going a slightly different route, on 11 March 2026, the German Federal Ministry of Economic Affairs proposed that petrol stations would only be allowed to increase prices once per day \citep{Handelsblatt2026}, whereas the number of price decreases would remain unlimited. No constraints on the magnitude of price changes would be imposed. Austria has been using this kind of regulation since 2011 \citep{Obradovits2014}.

A possible reason why petrol station prices currently change so much and often is that stations price-discriminate across consumers at different times of the day. In the morning, commuters and firms are more likely to be forced to buy petrol for work, so they are relatively price-insensitive. In the evening, there is little work-related traffic, so consumers are more price sensitive. This pattern explains why petrol station prices tend to be highest in the morning and lowest in the evening \citep{Bundeskartellamt2021}. Recent ADAC evidence based on more than 14,000 stations shows that this pattern remained very pronounced in 2025, with an average Super E10 spread of 12.5 cents per liter between the most expensive and cheapest moments of the day \citep{ADACDaily2025}. A January 2026 ADAC release still reports an evening-morning spread of roughly 13 cents, so the pattern does not appear to have vanished in more recent data \citep{ADACJanuary2026}.

Given this setting, the timing of the price reset in the new proposed policy matters for the overall price level. Stations will still want to have high prices in the morning to extract surplus from commuters and firms, while wanting to compete more fiercely in the evening on price in order to attract price-sensitive consumers. Hence, if the price reset were allowed only at the time when prices peak without restrictions, it might change price levels relatively little. On the other hand, allowing a price increase only in the early afternoon---when prices were decreasing absent constraints---would change the pricing among petrol stations considerably. For the regulator, this induces an optimization problem to impose constraints under which petrol station equilibrium behavior yields the lowest price level.

In order to determine which price reset-hour is optimal in reducing the price level throughout the day, I use a simple spatial-competition model to infer the share of price-sensitive consumers $\lambda_t$ for every hour $t$ of the day from observed hourly prices during weekdays. Then, given the inferred sequence of the price-sensitivity share, I use the spatial-equilibrium price paths---constrained by the fact that prices after the reset cannot rise for 24 hours---to determine which reset-hour would lead to the lowest prices across the day. Because petrol sale quantity data are not publicly available, I weight the different hours of the day by a proxy, the weekday passenger-car traffic frequencies from BASt (German Federal Highway and Transport Research Institute) federal-road counting stations. Under that traffic-weighted objective, the optimal reset-hour of the new policy is 11:00. This optimal reset-hour is close to the reset-hour used in the Austrian regulation, which allowed price increases only at noon.

The optimal reset-hour at 11:00 induces an equilibrium price path that is constant, i.e., prices do not change during the day. This is because for the reset at 11:00, stations cannot lower prices in the evening to compete for price-sensitive consumers without being locked into a low price in the morning, hurting their profits from price-insensitive commuters. Thus, the optimal price response to an 11:00 reset is a flat price path. Other reset-hours induce non-flat price paths, i.e., they induce intertemporal price discrimination. 

A theoretical result shows that, in this model, any reset-hour that induces a flat equilibrium price path minimizes the weighted average price, given stations' equilibrium best responses. The intuition is straightforward. Because marginal cost is constant throughout the day, firms' objective is aligned with raising the weighted average price, whereas the regulator’s objective is to lower it. The regulator therefore prefers exactly those reset-hours that firms dislike most.

A flat equilibrium price path means that, under the reset rule, stations cannot profit from intertemporal price discrimination. They are unable to improve on a single common price by charging relatively high prices in hours with less price-sensitive demand and lower prices in hours with more price-sensitive demand. By contrast, whenever the reset-hour-induced equilibrium path is non-flat, stations are still able to segment the day and exploit differences in demand composition across hours. That same flexibility raises the weighted average price. Hence, within the model, reset-hours that induce flat equilibrium prices deliver the lowest attainable average price level. This means that a secondary regulator objective---reducing price volatility---is perfectly consistent with lower price levels in this model framework.

One limitation of this analysis is that no hourly quantity data are readily available. While high frequency price data for every minute of every day are available, we do not know exactly how many units are sold at different hours. But this information naturally matters for welfare considerations and counterfactual pricing scenarios. The hourly passenger-car traffic frequencies as a transparent proxy to weight the different hours of the day are better than no weighting, but they are not perfect. Petrol is storable, and someone who always commutes in the morning may never buy in the morning. An additional limitation is that we do not observe how ``purchase-time-flexible'' consumers are. Even if we had perfect quantity data, these would be based on current daily price paths. If the new policy induces considerably different price paths, then consumers might switch to buying petrol at different hours of the day given these new price paths. Since we do not observe this counterfactual, this time-of-purchase substitution is hard to quantify and model. 

This article is meant as a policy contribution on the specific question of when to set the ``reset-hour'' of the new proposed policy. It takes no stance on whether this policy is sensible to begin with. Several articles in the literature have investigated theoretically and empirically whether the ``one reset per day'' rule increases or decreases prices, based on Austrian data where the policy has been in effect for a while. \citet{Obradovits2014} develops a model of Austrian-style gasoline price regulation and shows that restricting the timing of price adjustments can, in principle, backfire by softening competition and raising prices. \cite{DewenterHeimeshoffLueth2017} provide reduced-form empirical evidence on regulatory interventions in retail gasoline markets and conclude that the Austrian regime lowered fuel prices, in contrast to the weaker effects they document for Western Australia. \cite{BeckerPfeiferSchweikert2021} use a synthetic-control design and find that Austria’s Fuel Price Fixing Act led to a substantial short-run decline in retail fuel prices relative to the counterfactual. \citet{FasoulaSchweikert2020} studies Austrian evidence on regulated price-adjustment dynamics, finding partial evidence that the new regulation passes through input price changes faster. More recently, \cite{BernhardtBreiderhoffDewenter2025} revisit transparency regulation in European fuel markets and suggest that, in Austria, the initial price-reducing effect of the 2011 reform may have been partly offset over time by the effects of increased market transparency.

Beyond evaluations of this specific policy, a large industrial-organization literature shows that price discrimination under competition may either intensify or soften competition, depending on how segmentation changes firms' incentives \citep{ThisseVives1988,Holmes1989,ArmstrongVickers2001,Stole2007}. In gasoline markets, search frictions, switching costs, and dynamic pricing are known to matter \citep{Borenstein1991,BorensteinShepard1996,Lewis2011,Noel2007}. German station-level evidence also shows substantial and persistent heterogeneity in retail prices across outlets and locations \citep{HaucapHeimeshoffSiekmann2017}.

\section{Model}
\subsection{Environment}
There are two petrol stations, indexed by $i \in \{A,B\}$, with constant marginal cost $c>0$. Time is discrete and measured in hours, $t=0,1,\ldots,23$. In each hour, a unit mass of consumers arrives. A share $\lambda_t \in (0,1]$ is price-sensitive, and the remaining share $1-\lambda_t$ is price-insensitive. 

The price-sensitive consumers are distributed uniformly on a Hotelling line $[0,1]$. Station $A$ is located at $0$ and station $B$ at $1$. A price-sensitive consumer at location $x$ obtains utility
\[
u_A=\bar u-p_A-\tau x,\qquad u_B=\bar u-p_B-\tau(1-x),
\]
where $\tau>0$ measures travel or distance cost. The market is fully covered. Standard Hotelling algebra implies that the price-sensitive demand for station $i$ in hour $t$ is
\[
q_{it}^S(p_{it},p_{jt})=\lambda_t\left(\frac{1}{2}+\frac{p_{jt}-p_{it}}{2\tau}\right).
\]

The price-insensitive consumers are local. Each station has a local mass $(1-\lambda_t)/2$ in hour $t$. If the station charges $p_{it}\leq v$, all of them buy there; if $p_{it}>v$, they do not buy. Thus
\[
q_{it}^I(p_{it})=
\begin{cases}
\frac{1-\lambda_t}{2}, & p_{it}\leq v,\\
0, & p_{it}>v.
\end{cases}
\]
The parameter $v$ is large enough that the equilibrium prices in the calibration stay in the interior region, hence I do not consider corner solutions further. 

To capture the fact that some hours are quantitatively more important than others, let
\[
w_t>0,
\qquad
\sum_{t=0}^{23} w_t = 1,
\]
denote hour-specific demand weights in clock time. In the unrestricted hourly problem, this weight is a multiplicative scale factor in profits and therefore does not affect Proposition \ref{prop1}. It matters once one common price has to cover several hours, as in Propositions \ref{prop2} and \ref{prop3}.

\subsection{Unrestricted hourly pricing}
Under unrestricted pricing, stations can choose an independent price in every hour.

\begin{proposition}\label{prop1}
If the reservation constraint does not bind, the symmetric equilibrium price in hour $t$ is
\[
p_t^F=c+\frac{\tau}{\lambda_t}.
\]
\end{proposition}
\begin{proof}
See Appendix.
\end{proof}

The unrestricted price therefore rises when the share of price-sensitive consumers falls, as it lessens spatial-competition and prices are raised to extract more from the price-insensitive consumers. The observed intraday price cycle can be interpreted as a reduced-form path of $\lambda_t$.

\subsection{A single daily upward price reset}
Now suppose each station may increase its price only once per day, at a clock hour $r$. The number of price cuts remains unrestricted. Between two reset opportunities, the price path must therefore be weakly decreasing.

For the reset-rule analysis, it is convenient to introduce a new within-cycle time index
\[
k=0,1,\ldots,23,
\]
where $k=0$ is the reset-hour, $k=1$ is the next hour, and so on until $k=23$, the hour just before the next daily reset. Let $t_k(r)$ denote the corresponding clock hour. Thus $t_0(r)=r$, $t_1(r)$ is the next clock hour after $r$, etc.

Consider any consecutive block $K$ of ``within-cycle'' hours (determined by the reset-time) that must share one common price $p_K$. For example, if the two hours after the reset have the same price, then $k=0,1$ are part of the same block. Total profit in that block is
\[
\Pi_i^K=(p_i-c)\sum_{k\in K}w_{t_k(r)}\left[\lambda_{t_k(r)}\left(\frac{1}{2}+\frac{p_j-p_i}{2\tau}\right)+\frac{1-\lambda_{t_k(r)}}{2}\right]
=(p_i-c)\sum_{k\in K}w_{t_k(r)}\left[\frac{1}{2}+\lambda_{t_k(r)}\frac{p_j-p_i}{2\tau}\right].
\]

\begin{proposition}\label{prop2}
If the reservation constraint does not bind, and the price within a block has to be constant, then the symmetric equilibrium price for block $K$ is
\[
p_K^R=c+\frac{\tau}{\bar \lambda_K^w},
\qquad
\bar \lambda_K^w \equiv \frac{\sum_{k\in K}w_{t_k(r)}\lambda_{t_k(r)}}{\sum_{k\in K}w_{t_k(r)}}.
\]
\end{proposition}
\begin{proof}
See Appendix.
\end{proof}

The superscript ``R'' indicates restricted equilibrium prices, rather than flexible ``F''-prices as in Proposition \ref{prop1}. Now, under the new policy, where only one price increase is allowed at hour $r$, what is the new daily price path?

\begin{proposition}\label{prop3}
Fix a reset-hour $r$, and index within-cycle hours by $k=0,\dots,23$, where $k=0$ is the reset-hour. Let
\[
p_k^F(r)\equiv p_{t_k(r)}^F
\]
denote the unrestricted hourly equilibrium prices from Proposition \ref{prop1}. Under the once-per-day reset rule, the symmetric equilibrium path $\mathbf p^R(r)=\bigl(p_0^R(r),\dots,p_{23}^R(r)\bigr)$ is obtained as follows.

\medskip
\noindent
\textbf{Step 1.} Start from the 24 singleton blocks
\[
\{0\},\{1\},\dots,\{23\},
\]
with initial block prices
\[
p_0^F(r),p_1^F(r),\dots,p_{23}^F(r).
\]

\medskip
\noindent
\textbf{Step 2.} If the current block prices are weakly decreasing in within-cycle time, stop. Otherwise, whenever two adjacent blocks would imply an upward step in price, merge those two blocks into one larger consecutive block.

\medskip
\noindent
\textbf{Step 3.} For every merged block $K$, replace the prices of all hours in that block by the common block price
\[
p_K^R=c+\frac{\tau}{\bar\lambda_K^w},
\qquad
\bar\lambda_K^w\equiv
\frac{\sum_{k\in K} w_{t_k(r)}\lambda_{t_k(r)}}{\sum_{k\in K} w_{t_k(r)}}.
\]

\medskip
\noindent
\textbf{Step 4.} Repeat Steps 2 and 3 until no upward step remains. Then merge any adjacent blocks that have the same price. The unique endpoint is the symmetric equilibrium path $\mathbf p^R(r)$.

\medskip
\noindent
Thus, the equilibrium partition
\[
\mathcal K(r)=\{K_1(r),\dots,K_M(r)\}
\]
is the unique partition of the within-cycle hours into consecutive constant-price blocks produced by this iterative pooling procedure. For each block $K_m(r)$, let $p_m^R(r)$ denote its common equilibrium price, so that
\[
p_k^R(r)=p_m^R(r)
\qquad\text{for every } k\in K_m(r).
\]
Within each block the price is constant, while between any two adjacent blocks the equilibrium price drops strictly. Hence
\[
p_1^R(r)>p_2^R(r)>\cdots>p_M^R(r).
\]
\end{proposition}

\begin{proof}
See Appendix.
\end{proof}

The interesting thing about Proposition \ref{prop3} is that the equilibrium can be constructed iteratively. Start from the symmetric equilibrium price path over the day if there were no constraints (Proposition \ref{prop1}). If the hourly prices in that path respect monotonicity between resets, i.e., they do not increase except in the reset-hour, then we have the equilibrium price path under the new policy. But if the unconstrained price path does not respect monotonicity, take any two subsequent hours where the price increases (excluding the reset-hour where a price increase is allowed). Then merge these two hours into a block. The equilibrium price for that block---where prices have to be constant---is based on the demand-weighted average of the share of price-sensitive consumers in those two hours, using the corresponding hour weights $w_t$ (Proposition \ref{prop2}). Continue merging hours and blocks until the resulting block partition has a price path that respects monotonicity. Hence, the non-monotonicity between any two subsequent hours or blocks is removed. Proposition \ref{prop3} states that this price path found through iterative merging is a symmetric equilibrium.\footnote{The solution for the underlying problem comes from an older literature in statistics on isotonic regression problems that first showed that a quadratic optimization problem subject to a monotonicity constraint can be solved with the pool-adjacent violators algorithm (e.g., \cite{BarlowBartholomewBremnerBrunk1972}). This is the algorithm described in Proposition \ref{prop3}.}

Hence, Proposition \ref{prop2} answers the question ``What price arises in equilibrium if I already know which hours of the day must share one price?''. Proposition \ref{prop3} answers the question ``Which hours in equilibrium have to share one price once the reset-hour $r$ is fixed?''. 

Once Proposition \ref{prop3} delivers the constrained equilibrium price path $\{p_k^R(r)\}_{k=0}^{23}$ for reset-hour $r$, the regulator's objective is
\[
\bar p^w(r)=\sum_{k=0}^{23} w_{t_k(r)} p_k^R(r),
\]
that is, the price at every hour is weighted by the quantity purchased, with the caveats of imperfect quantity proxy and the assumption of no time-substitution under new price paths as discussed above.

\section{Traffic-Weighted Calibration}
\subsection{Weekday station-level E5 prices}
I use historical station-level fuel price-change files from Tankerkönig's data portal \citep{TankerkoenigData2026}. The raw files record time-stamped price changes rather than a complete panel of prices at every instant. For the price profile I keep the 15 weekdays from February 2 to February 20, 2026. The three preceding Sundays -- February 1, 8, and 15 -- are used only to initialize Monday midnight prices; without them we would be missing the station prices each Monday until the first price change that day. I choose this February window for two reasons. First, it is recent and close to the March 2026 German policy proposal. Second, it ends before late February when fuel prices were publicly described as unusually high amid Middle East war risk \citep{ADACWarRisk2026}. The sample is therefore a recent pre-shock window rather than a war-distorted one.

I use E5 rather than E10 because it is the more widely used petrol grade and therefore the more relevant benchmark for routine commuter refueling. I focus on weekdays because commuter traffic and business activity is concentrated there and purchase timing is less flexible than on weekends.\footnote{The exercise could be repeated for weekend days, and the optimal reset-hour on weekends may very well be different, given the vastly different commuting patterns.} Transforming the raw Tankerkönig files into hourly prices requires two steps. First, I treat each station's E5 price as piecewise constant between two successive recorded E5 changes. Second, for each station-hour I compute a time-weighted average price, so a price that lasts 45 minutes within the hour receives three times the weight of a price that lasts 15 minutes. To avoid thin Monday morning coverage, I carry each station's Sunday closing price into the following Monday until the first Monday change. Prices from weekdays are also carried over to the next weekday until there is a price update. I then pool these station-hour prices across all observed station-seconds in the same clock hour. With the Sunday carryover, the average observed station count per hour lies between about 14,041 and 14,178, so the resulting profile is close to balanced over the day.

Figure \ref{fig:profile} shows the pooled weekday profile. The pattern is intuitive and very similar to the well-known German intraday cycle: prices are relatively high overnight, jump sharply around 07:00, and then decline in several waves toward an evening trough. In the hourly aggregation, the highest price is 187.92 cents per liter at 07:00 and the lowest is 178.13 cents at 19:00. The 9.79-cent gap across hourly averages is smaller than the corresponding moment-to-moment spread because full-hour averaging smooths short-lived peaks and troughs.

\begin{figure}[!htbp]
\centering
\includegraphics[width=0.92\textwidth]{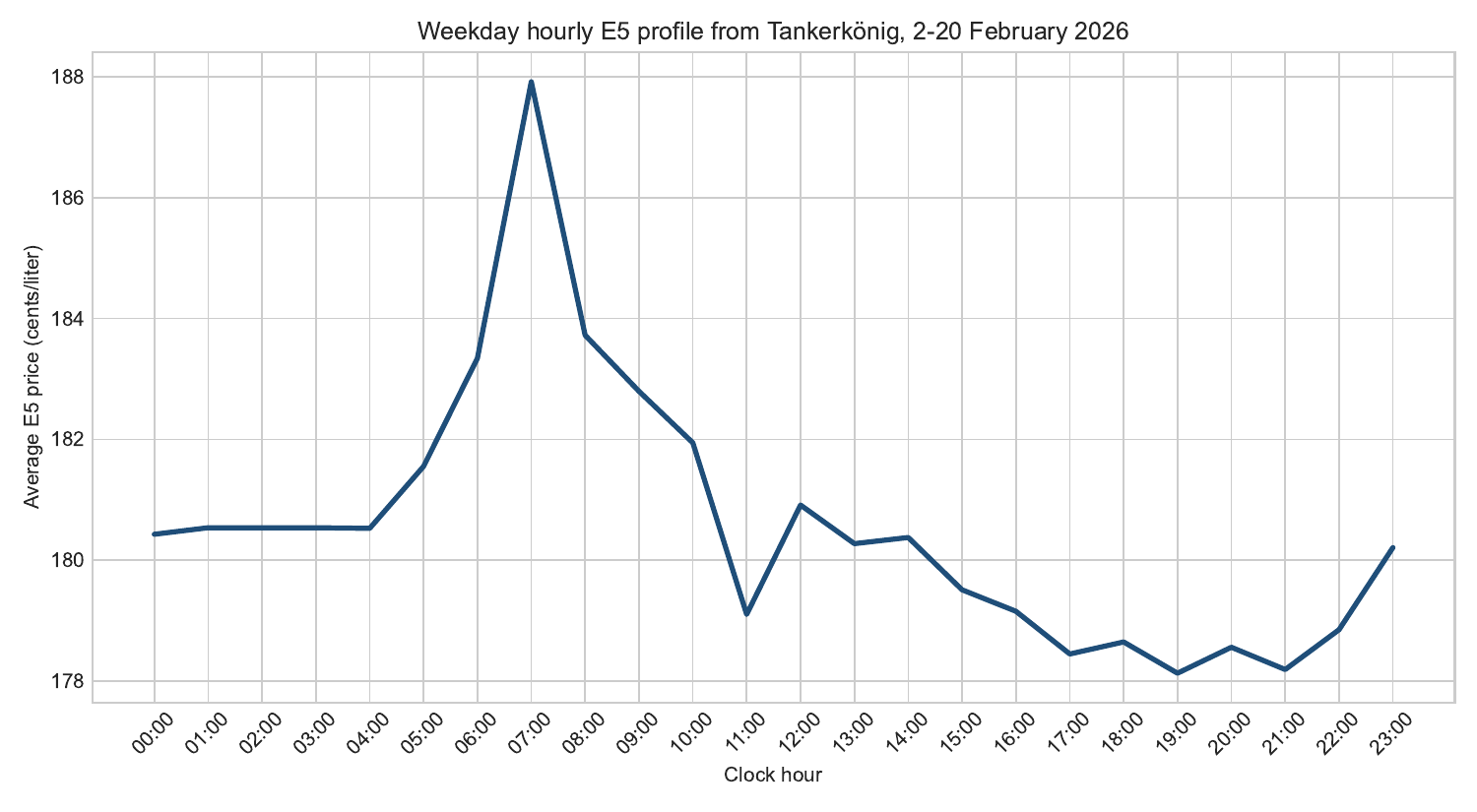}
\caption{Weekday intraday E5 profile from Tankerkönig}
\label{fig:profile}
\begin{minipage}{0.92\textwidth}
\footnotesize Notes: The figure plots time-weighted station-hour averages for E5 prices from Tankerkönig over the weekdays from February 2 to February 20, 2026. The preceding Sundays are used only to initialize Monday carryover states and are not included in the weekday profile.
\end{minipage}
\end{figure}

\subsection{Inferring the hourly share of price-sensitive consumers}
If the reservation constraint does not bind, Proposition \ref{prop1} implies
\[
\lambda_t=\frac{\tau}{p_t-c}.
\]
To map the price data into shares, I normalize the cheapest hour to $\lambda_t=1$. This implies
\[
\widehat{\lambda}_t=\frac{p_{\min}-c}{p_t-c}.
\]
The level of $c$ is a calibration choice rather than an estimate, so I use a baseline value of $c=1.50$ euro per liter and report sensitivity to different values of $c\in\{1.45,1.50,1.55\}$.\footnote{Because the calibration normalizes the cheapest hour to $\widehat{\lambda}_{\min}=1$, the model implies $\tau=p_{\min}-c$. Thus the sensitivity exercise over $c\in\{1.45,1.50,1.55\}$ also varies the implied Hotelling transport-cost parameter $\tau$ one-for-one. In other words, $\tau$ is not calibrated separately; it is pinned down mechanically by the normalization.}

Figure \ref{fig:lambda} plots the resulting hourly shares of price-sensitive consumers. In the baseline calibration, the inferred share of price-sensitive consumers is about 0.742 at 07:00, rises to roughly 0.967 by 11:00, and reaches one at 19:00. The sensitivity band is economically modest: the hour of lowest price sensitivity remains the morning peak, and the hour of greatest price sensitivity remains the evening trough for all three marginal cost values.

\begin{figure}[!htbp]
\centering
\includegraphics[width=0.92\textwidth]{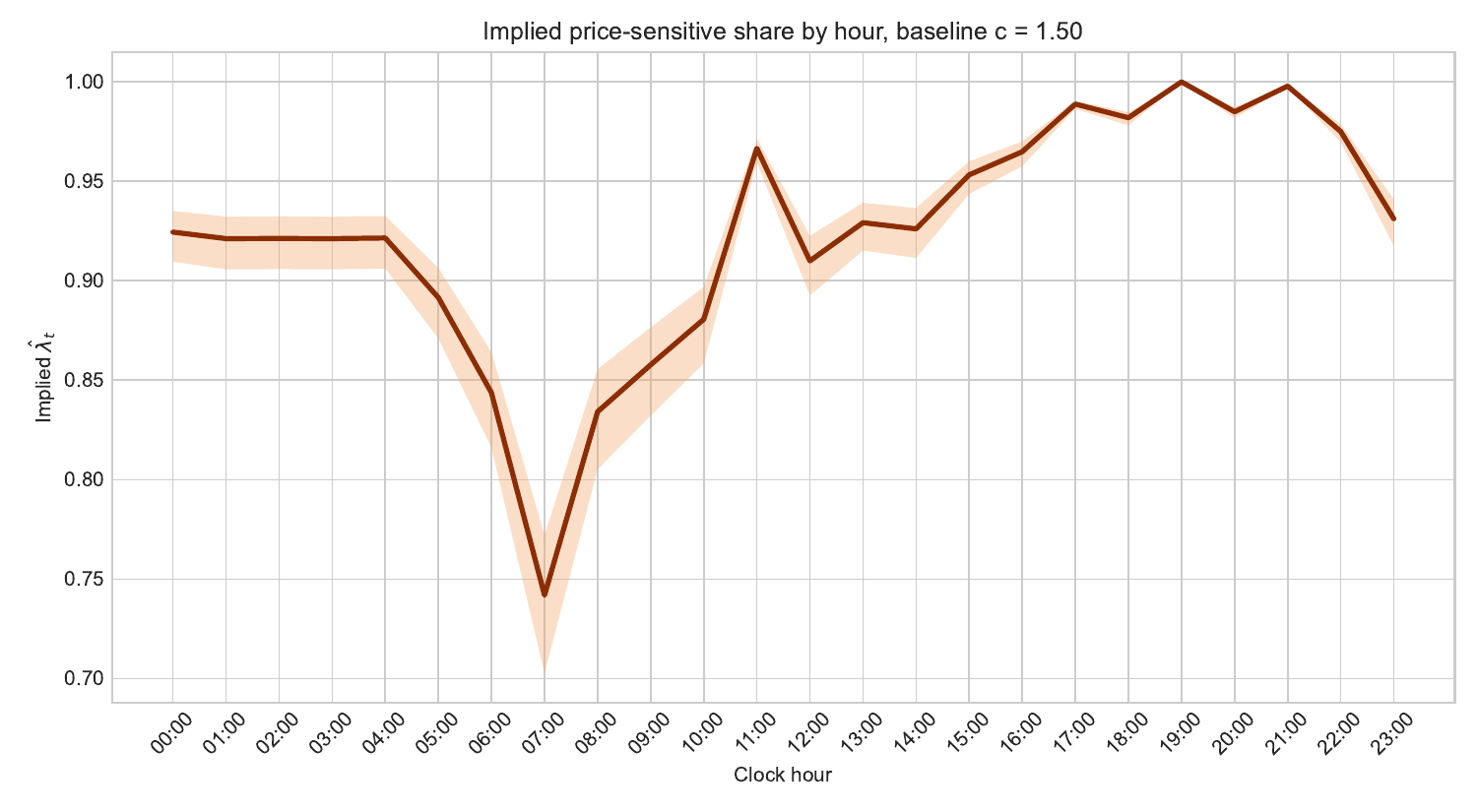}
\caption{Inferred share of price-sensitive consumers by hour}
\label{fig:lambda}
\begin{minipage}{0.92\textwidth}
\footnotesize Notes: The solid line uses the baseline cost calibration $c=1.50$ euro per liter. The shaded band shows the range obtained from $c=1.45$ and $c=1.55$. The underlying hourly profile is the pooled Tankerkönig weekday E5 price profile for February 2--20, 2026, and the cheapest hour is normalized to $\lambda_t=1$.
\end{minipage}
\end{figure}

\subsection{Traffic weights from BASt passenger-car counts}
Hourly fuel quantities are not publicly available, but the German BASt publishes hourly raw traffic counts from permanent counting stations on federal trunk roads \citep{BAStTraffic2026,BAStTrafficFormat2026}. I use the February 2025 Bundesstraßen archive because the 2026 data is not yet available, and it matches the month/season of the Tankerkönig price sample. Several design choices matter. First, I use Bundesstraßen (federal roads) rather than Autobahnen (motorways), because federal roads are closer to routine local and commuter refueling, whereas motorway counts include much more long-distance traffic and freight. Second, I use passenger cars (Pkw) rather than total vehicles, because E5 is a petrol grade and truck traffic is disproportionately diesel. Third, I keep weekdays only, to match the weekday focus of the price sample.

The BASt raw file has one row per counting station, day, hour, giving counts for different vehicle classes for both directions separately. To be conservative, I keep only counts where passenger-car measure is unflagged in the corresponding status column. Reassuringly, the share of unflagged measurements is within one percentage point across all hours, so there is no meaningful hour-specific bias due to missing data. Aggregating across the 870 Bundesstraßen counting stations yields hourly passenger-car totals $Q_t$, and I define the empirical weights as
\[
\widehat w_t=\frac{Q_t}{\sum_{s=0}^{23} Q_s}.
\]
These weights enter twice in the calibration. First, they determine the weighted block averages of $\lambda_t$ that govern firms' constrained equilibrium prices under Proposition \ref{prop3}. Second, they weight the regulator's average-price objective. The resulting weights are strongly concentrated in the commuter day. In the February weekday sample, 92.8\% of total weight lies between 06:00 and 21:00, 07:00 receives 7.53\% of total weight, the maximum occurs at 16:00 with 8.61\%, and 02:00 receives only 0.15\%.

\begin{figure}[!htbp]
\centering
\includegraphics[width=0.92\textwidth]{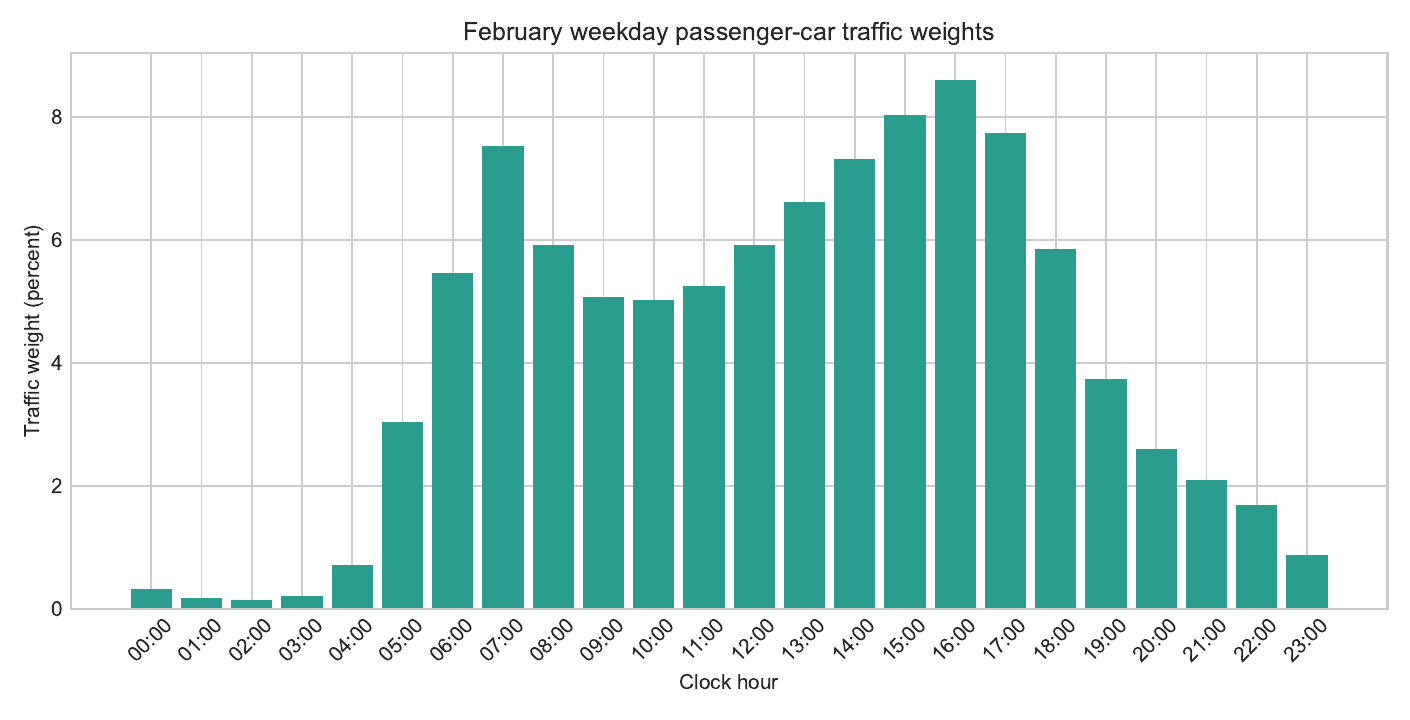}
\caption{Traffic weights by hour}
\label{fig:weights}
\begin{minipage}{0.92\textwidth}
\footnotesize Notes: The bars show $\widehat w_t$, the share of total February 2025 weekday passenger-car traffic that falls in clock hour $t$ on BASt Bundesstraßen counting stations. Only unflagged direction counts are used.
\end{minipage}
\end{figure}

\begin{table}[!htbp]
\centering
\caption{Hourly prices, inferred price sensitivity, and traffic weights}
\label{tab:hourlydata}
\small
\begin{tabular}{cccccccccccc}
\toprule
Hour & Price & $\lambda_t$ & Weight & Hour & Price & $\lambda_t$ & Weight & Hour & Price & $\lambda_t$ & Weight \\
\midrule
00:00 & 1.804 & 0.924 & 0.33 & 08:00 & 1.837 & 0.834 & 5.92 & 16:00 & 1.792 & 0.965 & 8.61 \\
01:00 & 1.805 & 0.921 & 0.19 & 09:00 & 1.828 & 0.858 & 5.08 & 17:00 & 1.784 & 0.989 & 7.74 \\
02:00 & 1.805 & 0.921 & 0.15 & 10:00 & 1.819 & 0.881 & 5.02 & 18:00 & 1.786 & 0.982 & 5.85 \\
03:00 & 1.805 & 0.921 & 0.21 & 11:00 & 1.791 & 0.967 & 5.26 & 19:00 & 1.781 & 1.000 & 3.73 \\
04:00 & 1.805 & 0.921 & 0.72 & 12:00 & 1.809 & 0.910 & 5.93 & 20:00 & 1.786 & 0.985 & 2.60 \\
05:00 & 1.816 & 0.892 & 3.04 & 13:00 & 1.803 & 0.929 & 6.63 & 21:00 & 1.782 & 0.998 & 2.10 \\
06:00 & 1.833 & 0.844 & 5.47 & 14:00 & 1.804 & 0.926 & 7.31 & 22:00 & 1.789 & 0.975 & 1.69 \\
07:00 & 1.879 & 0.742 & 7.53 & 15:00 & 1.795 & 0.953 & 8.03 & 23:00 & 1.802 & 0.931 & 0.88 \\
\bottomrule
\end{tabular}
\\ [2mm]
\begin{minipage}{0.94\textwidth}
\footnotesize Notes: Prices are time-weighted E5 station-hour averages from Tankerkönig over the weekdays from February 2 to February 20, 2026. The baseline share uses $c=1.50$ euro per liter and normalizes the cheapest hour to $\lambda_t=1$. The traffic weights are February 2025 weekday passenger-car shares from BASt Bundesstraßen counting stations, using only unflagged direction counts.
\end{minipage}
\end{table}

\section{When Should the Single Price Increase Be Allowed?}
I now evaluate every candidate reset-hour $r=0,\ldots,23$. For each $r$, I work in within-cycle time $k=0,\ldots,23$, where $k=0$ is the reset-hour itself and $k=23$ is the hour just before the next daily reset.

First, I form the sequence
\[
\widehat{\lambda}_0(r),\ \widehat{\lambda}_1(r),\ \ldots,\ \widehat{\lambda}_{23}(r),
\qquad
\widehat{\lambda}_k(r)\equiv \widehat{\lambda}_{t_k(r)}.
\]
This is just the observed day reordered so that the cycle begins at the candidate reset-hour. I also form the corresponding weight sequence
\[
\widehat{w}_0(r),\ \widehat{w}_1(r),\ \ldots,\ \widehat{w}_{23}(r),
\qquad
\widehat{w}_k(r)\equiv \widehat{w}_{t_k(r)}.
\]

Second, I impose the once-per-day reset constraint by requiring the price path to be weakly decreasing in the within-cycle time variable $k$. Equivalently, the demand-weighted block-average values of $\widehat{\lambda}_k(r)$ must be weakly increasing. I therefore pool adjacent violating hours until I obtain the unique partition from Proposition \ref{prop3}, using $\widehat{w}_k(r)$ as the observation weights. In practice, this is a weighted pool-adjacent-violators algorithm.

Third, for every resulting block $K_m(r)$ I assign the equilibrium block price
\[
p_m^R(r)=c+\frac{\tau}{\bar\lambda_m^w(r)},
\qquad
\bar\lambda_m^w(r)\equiv
\frac{\sum_{k\in K_m(r)}\widehat{w}_k(r)\widehat{\lambda}_k(r)}
{\sum_{k\in K_m(r)}\widehat{w}_k(r)},
\]
since the maintained calibration assumes that the reservation price $v$ of price-insensitive consumers is not hit in practice. This yields the full within-cycle equilibrium path
\[
\{p_0^R(r),p_1^R(r),\ldots,p_{23}^R(r)\}.
\]
Fourth, I compute the traffic-weighted objective
\[
\bar p^{\,w}(r)=\sum_{k=0}^{23} \widehat w_k(r) p_k^R(r).
\]
The optimal reset-hour is the minimizer of this weighted average price criterion.

Table \ref{tab:besthours} summarizes the best-performing hours. Under the baseline calibration of $c=1.5$, the best hour for the reset is 11:00. The unrestricted traffic-weighted average price is 180.874 cents per liter, while the best regulated traffic-weighted average is 180.674 cents, a reduction of 0.199 cents per liter. The ranking is unchanged across the cost grid $c\in\{1.45,1.50,1.55\}$. Note that the second to fifth ranked reset times all lead to an equilibrium price path with two blocks---i.e., a price drop during the day---whereas the best reset time yields a constant equilibrium price path during the day.

\begin{table}[!htbp]
\centering
\caption{Best candidate hours under traffic weighting}
\label{tab:besthours}
\small
\begin{tabular}{ccccc}
\toprule
Rank & Hour & Weighted avg. price (ct/l) & Saving (ct/l) & Blocks \\
\midrule
1 & 11:00 & 180.674 & 0.199 & 1 \\
2 & 13:00 & 180.676 & 0.197 & 2 \\
3 & 14:00 & 180.677 & 0.197 & 2 \\
4 & 15:00 & 180.677 & 0.197 & 2 \\
5 & 10:00 & 180.677 & 0.196 & 2 \\
\bottomrule
\end{tabular}
\\ [2mm]
\begin{minipage}{0.90\textwidth}
\footnotesize Notes: Candidate hours are evaluated in the baseline calibration $c=1.50$ euro per liter. The objective is the traffic-weighted average price $\sum_t w_t p_t^R(r)$, where the weights come from February 2025 weekday passenger-car counts on BASt Bundesstraßen counting stations. The constrained block prices also use these hour weights through traffic-weighted block averages of $\lambda_t$. Equal-hour averages are reported separately in the machine-readable output.
\end{minipage}
\end{table}

Why are prices flat at the 11:00 reset-hour? That flat equilibrium price is above the evening trough, but far below the 07:00 commuter peak. From the observed price path across the day---interpreted as unconstrained equilibrium prices---we know petrol stations want to lower prices in the evening to attract more price-sensitive consumers, while ramping up prices in the morning to extract profits from price-insensitive commuters. But with an 11:00 reset, petrol stations cannot at the same time decrease prices in the evening and then increase prices in the morning, because the price increase is only allowed mid-day after the morning rush. Thus, because of the concave profit objective, petrol stations prefer to set a constant intermediate price across the day, rather than a price close to the high morning optimum or close to the low evening optimum. 

Figure \ref{fig:paths} plots the new equilibrium price path under the new policy with the optimal reset-hour against the old equilibrium. Relative to unrestricted pricing, the best weighted rule replaces the intraday price cycle seen now with a completely flat price path set by petrol stations. The gain from the policy therefore comes from cutting the high-price commuter hours by more than it raises the low-price evening and night hours.

\begin{figure}[!htbp]
\centering
\includegraphics[width=0.92\textwidth]{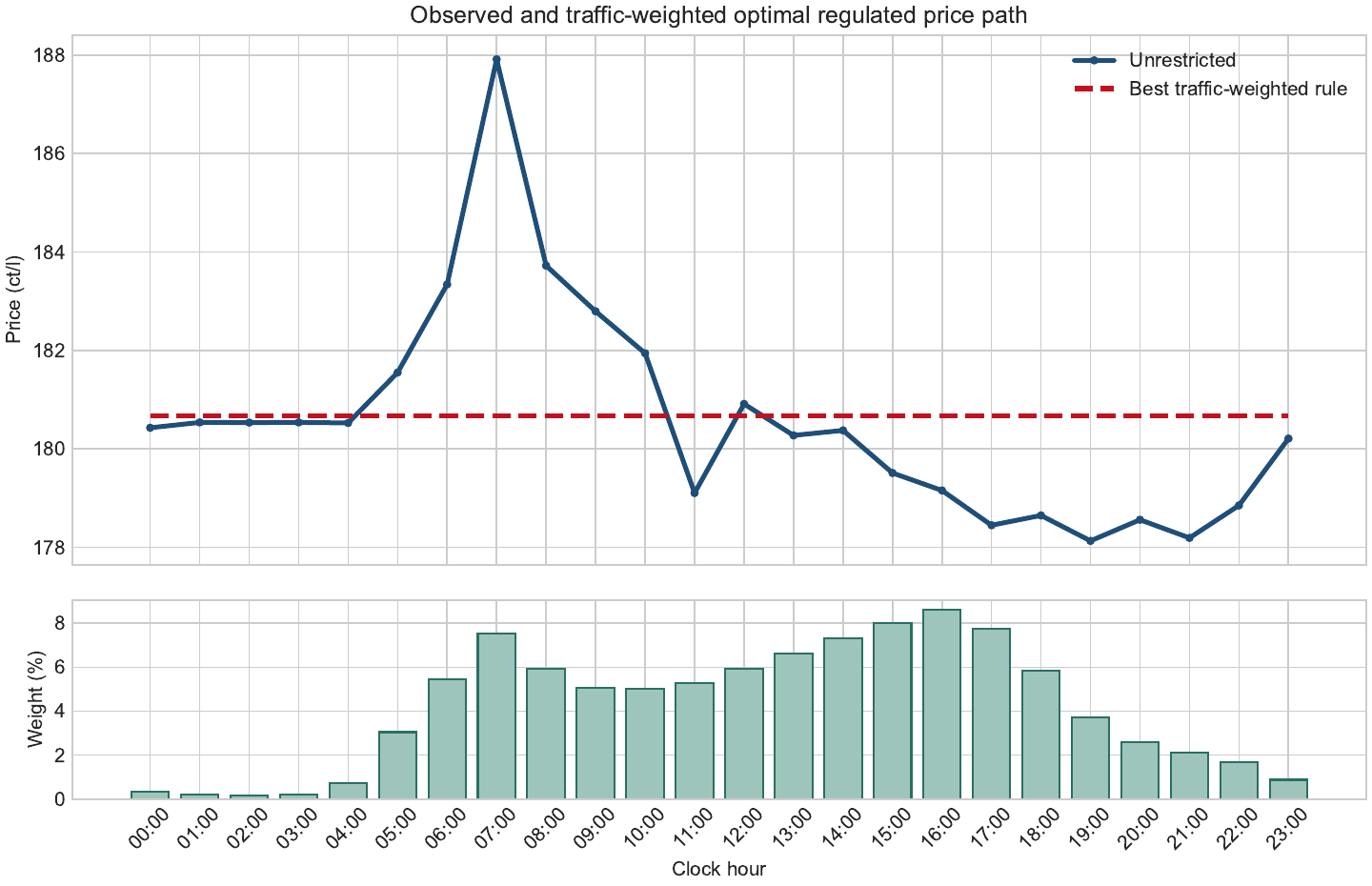}
\caption{Observed and traffic-weighted optimal regulated price path}
\label{fig:paths}
\begin{minipage}{0.92\textwidth}
\footnotesize Notes: The blue line is the pooled Tankerkönig weekday E5 profile for February 2--20, 2026. The dashed red line is the regulated path for the reset-hour that minimizes the traffic-weighted average price. The bars show the hourly traffic weights $\widehat w_t$ from February 2025 weekday passenger-car counts on BASt Bundesstraßen counting stations.
\end{minipage}
\end{figure}

Within the model, the new policy predicts lower traffic-weighted prices throughout the day. But there are winners and losers from this policy change. Price-sensitive consumers that usually bought in the evening would pay more under the new policy, whereas commuters and firms in the morning pay less. And while some might welcome fewer price changes during the day, making prices at the petrol station more predictable, there are also economic arguments in favor of ``peak pricing''. If there is a constant price throughout the day for petrol, as predicted by the model for the best reset-hour under the new policy, then there might be more crowding in peak hours as price-sensitive consumers no longer have an incentive to wait for less busy evenings.

The second best reset-hour according to Table \ref{tab:besthours} at 13:00 has two price blocks. The first block runs from 13:00 until 10:00 the next day at slightly higher prices, which includes the morning commuter rush. And the second block runs from 11:00 to 12:00 at lower prices to benefit from competition for the higher share of price-sensitive consumers around noon. This shows that the choice of the reset-hour affects both price volatility and price levels.

\section{A General Theoretical Result: Flat Prices Are Optimal}

In this setting, the regulator cannot dictate prices; it can only choose the reset-hour, thereby imposing different constraints on petrol stations' price setting behavior.

It is no coincidence that the reset-hour which induces a flat equilibrium price path yields the lowest average price. Both the policy objective (lowest average weighted price) and the petrol station objective (highest weighted average price, given constant cost $c$ during the day) are the same but with opposite signs. So a reset-hour is attractive for the regulator exactly when it is unattractive for the stations. A flat price path is always feasible for petrol stations under the once-per-day reset rule. Therefore, whenever the equilibrium path is non-flat, the equilibrium price changes reveal that stations strictly prefer to use the available within-day flexibility in order to segment hours with different demand composition. That is, they engage in intertemporal price discrimination to keep prices relatively high in blocks with a lower average share of price-sensitive consumers and cut prices only in blocks with a higher average share of price-sensitive consumers.


The following proposition shows that this intertemporal segmentation necessarily raises the weighted average price level relative to a flat path. Formally, each block price is proportional to the inverse of the block-average demand sensitivity, $1/\bar\lambda_m^w(r)$. Since $1/x$ is strictly convex in $x>0$, by Jensen's inequality, splitting the day into several blocks with different average values of $\lambda_t$ increases the weighted average of block prices relative to assigning one common price to the whole day. Put differently, variation in block-level demand sensitivity is valuable to firms because it lets them tailor markups more closely to the composition of demand at different times. But that same flexibility is exactly what the regulator wants to eliminate. Hence, within this model, a reset-hour that produces a flat equilibrium path is optimal among all reset-hours, and any reset-hour that leaves scope for intraday price discrimination must perform strictly worse in terms of weighted average price levels. This result also implies that a secondary regulator objective---reducing price volatility---is perfectly consistent with lower price levels in this model framework.

\begin{proposition}\label{prop:flat_path_dominance}
Suppose the reservation constraint does not bind on any equilibrium block, so Proposition \ref{prop3} yields equilibrium prices
\[
p_m^R(r)=c+\frac{\tau}{\bar\lambda_m^w(r)}
\]
for every reset hour $r$ and every block $K_m(r)$. Define block weights
\[
W_m(r)\equiv \sum_{k\in K_m(r)} w_{t_k(r)},
\qquad
\bar\lambda_m^w(r)\equiv
\frac{\sum_{k\in K_m(r)} w_{t_k(r)}\lambda_{t_k(r)}}{W_m(r)},
\]
and define the day-wide weighted mean of the price-sensitive consumer share by
\[
\Lambda^w \equiv \sum_{t=0}^{23} w_t \lambda_t.
\]
Then, for every reset hour $r$, average equilibrium prices are related to the lower bound of average equilibrium prices via
\[
\bar p^w(r)\equiv \sum_{m=1}^{M(r)} W_m(r)p_m^R(r)
= c+\tau \sum_{m=1}^{M(r)} \frac{W_m(r)}{\bar\lambda_m^w(r)}
\ge c+\frac{\tau}{\Lambda^w}.
\]
Equality holds if and only if the equilibrium price path is flat, or equivalently, if the number of blocks is one $(M(r)=1)$. Hence any reset hour that induces a flat equilibrium price path attains the lowest possible weighted average price, and any reset hour that induces a non-flat equilibrium price path yields a strictly higher weighted average price.
\end{proposition}
\begin{proof}
See Appendix.
\end{proof}

\section{Conclusion}
The timing of a once-per-day price increase rule might matter a lot for prices and consumer welfare. In a two-station model with time-varying demand composition, unrestricted hourly pricing is governed by the share of price-sensitive consumers, $\lambda_t$. If prices have to be constant within a block of several hours, the equilibrium price in that block depends on the \emph{quantity-weighted average} of that share within the block. This turns the policy problem into one of choosing the appropriate reset-hour for the single allowed price increase, each of which leads endogenously to a different equilibrium price path under the newly proposed upward-pricing constraint. By choosing the right reset-hour, the induced equilibrium price path over the day may exhibit lower or higher price levels.

Using weekday Tankerkönig station-level E5 petrol price data and BASt weekday passenger-car traffic weights for Bundesstraßen, I infer that price sensitivity is lowest around 07:00 and highest in the early evening. Evaluating every candidate reset-hour reveals that the best reset-hour under the traffic-weighted objective is 11:00, where equilibrium prices are constant across the day. A theoretical result shows that, in this model framework, inducing constant prices is in fact the best the regulator can do with the single reset-hour regulation.

This analysis might be extended if more precise quantity-by-hour data was available, in order to replace the traffic proxy with actual sales volumes. The underlying model could also be extended by using location data for every petrol station and determining a ``competition index'' for every station. This kind of heterogeneity in the model would make better price predictions for specific petrol stations, by taking into account just how much spatial competition they face given the distance to their competitors, but it is not clear whether it would change the aggregate lessons in this article. For that, the quantity data would seem more important.

\bibliography{petrol_policy_references_010}

\appendix
\section{Proofs}

\begin{proof}[Proof of Proposition \ref{prop1}]
In hour $t$, station $i$ solves
\[
\max_{p_{it}} \ w_t(p_{it}-c)\left[\lambda_t\left(\frac{1}{2}+\frac{p_{jt}-p_{it}}{2\tau}\right)+\frac{1-\lambda_t}{2}\right]
\]
for prices below $v$. The objective is strictly concave because the second derivative is $-w_t\lambda_t/\tau<0$. The first-order condition is
\[
\lambda_t\left(\frac{1}{2}+\frac{p_{jt}-p_{it}}{2\tau}\right)+\frac{1-\lambda_t}{2}-\lambda_t\frac{p_{it}-c}{2\tau}=0.
\]
where the positive factor $w_t$ cancels out. Imposing symmetry, $p_{it}=p_{jt}=p_t$, yields
\[
\frac{1}{2}-\lambda_t\frac{p_t-c}{2\tau}=0,
\]
which gives $p_t=c+\tau/\lambda_t$. 
\end{proof}

\begin{proof}[Proof of Proposition \ref{prop2}]
For a common price on block $K$, station $i$ solves
\[
\max_{p_i}\ (p_i-c)\sum_{k\in K}w_{t_k(r)}\left[\frac{1}{2}+\lambda_{t_k(r)}\frac{p_j-p_i}{2\tau}\right].
\]
The objective is strictly concave. The first-order condition is
\[
\frac{1}{2}\sum_{k\in K}w_{t_k(r)}+\frac{p_j-p_i}{2\tau}\sum_{k\in K}w_{t_k(r)}\lambda_{t_k(r)}-\frac{p_i-c}{2\tau}\sum_{k\in K}w_{t_k(r)}\lambda_{t_k(r)}=0.
\]
Imposing symmetry gives
\[
\frac{1}{2}\sum_{k\in K}w_{t_k(r)}-\frac{p_K-c}{2\tau}\sum_{k\in K}w_{t_k(r)}\lambda_{t_k(r)}=0,
\]
which implies
\[
p_K=c+\tau\frac{\sum_{k\in K}w_{t_k(r)}}{\sum_{k\in K}w_{t_k(r)}\lambda_{t_k(r)}}=c+\frac{\tau}{\bar \lambda_K^w}.\qedhere
\]
\end{proof}

\begin{proof}[Proof of Proposition \ref{prop3}]
Write $\lambda_k\equiv \lambda_{t_k(r)}$ and $w_k\equiv w_{t_k(r)}$, and first ignore the reservation cap $v$. For a given symmetric candidate path $p=(p_0,\dots,p_{23})$, station $i$'s payoff is
\[
\sum_{k=0}^{23} w_k (p_{ik}-c)\left[\frac{1}{2}+\lambda_k\frac{p_k-p_{ik}}{2\tau}\right],
\]
subject to the monotonicity constraints $p_{ik}\le p_{i,k-1}$ for $k=1,\dots,23$. 
For fixed rival $\{p_k\}$, the objective is strictly concave in $\mathbf{p}_i$, because the Hessian is diagonal with entries $-w_k\lambda_k/\tau<0$, and the feasible set is convex. Hence the best response is characterized by the KKT conditions. At a symmetric candidate path, $p_{ik}=p_k$ for all $k$. After imposing symmetry, multiplying the stationarity conditions of station $i$'s problem by the positive constant $2$ and leaving feasibility and complementary slackness unchanged, one obtains exactly the KKT system of this single optimization problem,
\begin{equation}
\max_{\mathbf p\in \mathbb R_+^{24}} \sum_{k=0}^{23}\left[w_k(p_k-c)-\frac{w_k\lambda_k}{2\tau}(p_k-c)^2\right]
\quad\text{s.t.}\quad p_k\le p_{k-1}\ \text{for}\ k=1,\dots,23.
\label{eq:alt_prop3_program}
\end{equation}
The objective is strictly concave and the feasible set is convex, so \eqref{eq:alt_prop3_program} has a unique maximizer. Now let
\[
p_k^F(r)=c+\frac{\tau}{\lambda_k}
\]
denote the unconstrained hourly optimum from Proposition \ref{prop1}. Completing the square in \eqref{eq:alt_prop3_program} gives
\begin{equation*}
\begin{split}
w_k(p_k-c)-\frac{w_k\lambda_k}{2\tau}(p_k-c)^2
=-\frac{w_k\lambda_k}{2\tau}\left[(p_k-c)^2-\frac{2\tau}{\lambda_k}(p_k-c)\right]\\
=-\frac{w_k\lambda_k}{2\tau}\left[\left(p_k-c-\frac{\tau}{\lambda_k}\right)^2-\left(\frac{\tau}{\lambda_k}\right)^2\right]
=\frac{w_k\lambda_k}{2\tau}\left(\frac{\tau}{\lambda_k}\right)^2-\frac{w_k\lambda_k}{2\tau}\bigl(p_k-p_k^F(r)\bigr)^2.
\end{split}
\end{equation*}
Hence \eqref{eq:alt_prop3_program} is equivalent to
\begin{equation}
\min_{\mathbf p} \sum_{k=0}^{23}\frac{w_k\lambda_k}{2\tau}\bigl(p_k-p_k^F(r)\bigr)^2
\quad\text{s.t.}\quad p_0\ge p_1\ge \cdots \ge p_{23}.
\label{eq:alt_prop3_isotonic}
\end{equation}
So the constrained equilibrium path is the unique weakly decreasing path that lies closest to the unconstrained price path in weighted quadratic loss. Equation \eqref{eq:alt_prop3_isotonic} is a weighted isotonic regression problem with fitted values $p_k$, data $p_k^F(r)$, and weights
\[
a_k\equiv \frac{w_k\lambda_k}{2\tau}>0.
\]
By the standard weighted pool-adjacent-violators theorem, its unique solution is obtained by starting from the singleton blocks $\{0\},\dots,\{23\}$ and repeatedly merging adjacent blocks whenever their fitted values would violate monotonicity, replacing them by their weighted average; see, for example, 
\citet{BarlowBartholomewBremnerBrunk1972}. Therefore the terminal partition produced by that iterative pooling procedure is the unique equilibrium partition, and the corresponding fitted values form the unique constrained equilibrium path.

For any pooled block $K$, the fitted value is the weighted average
\[
p_K^R
=
\frac{\sum_{k\in K} a_k p_k^F(r)}{\sum_{k\in K} a_k}.
\]
Substituting $a_k=w_k\lambda_k/(2\tau)$ and $p_k^F(r)=c+\tau/\lambda_k$ gives
\[
p_K^R
=
\frac{\sum_{k\in K}\frac{w_k\lambda_k}{2\tau}\left(c+\frac{\tau}{\lambda_k}\right)}
{\sum_{k\in K}\frac{w_k\lambda_k}{2\tau}}
=
c+\tau\frac{\sum_{k\in K}w_k}{\sum_{k\in K}w_k\lambda_k}
=
c+\frac{\tau}{\bar\lambda_K^w}.
\]
Hence the unique equilibrium path is exactly the path described in the proposition. If the unrestricted hourly prices already satisfy
\[
p_0^F(r)\ge p_1^F(r)\ge \cdots \ge p_{23}^F(r),
\]
then no pooling occurs and every block is a singleton.
\qedhere 

\end{proof}

\begin{proof}[Proof of Proposition \ref{prop:flat_path_dominance}]
Fix a reset hour $r$, and let $\mathcal K(r)=\{K_1(r),\dots,K_{M(r)}(r)\}$ denote the maximal constant-price blocks from Proposition \ref{prop3}. Because the path is constant on each block,
\[
\bar p^w(r)
=\sum_{m=1}^{M(r)} \sum_{k\in K_m(r)} w_{t_k(r)} p_m^R(r)
=\sum_{m=1}^{M(r)} W_m(r)p_m^R(r).
\]
Using the interior block-pricing formula,
\[
\bar p^w(r)
=\sum_{m=1}^{M(r)} W_m(r)\left(c+\frac{\tau}{\bar\lambda_m^w(r)}\right)
= c\sum_{m=1}^{M(r)} W_m(r)+\tau \sum_{m=1}^{M(r)} \frac{W_m(r)}{\bar\lambda_m^w(r)}.
\]
Since the blocks partition the day, $\sum_{m=1}^{M(r)} W_m(r)=\sum_{t=0}^{23} w_t=1$, so
\[
\bar p^w(r)=c+\tau \sum_{m=1}^{M(r)} \frac{W_m(r)}{\bar\lambda_m^w(r)}.
\]

Next,
\[
\sum_{m=1}^{M(r)} W_m(r)\bar\lambda_m^w(r)
=\sum_{m=1}^{M(r)} \sum_{k\in K_m(r)} w_{t_k(r)}\lambda_{t_k(r)}
=\sum_{k=0}^{23} w_{t_k(r)}\lambda_{t_k(r)}
=\sum_{t=0}^{23} w_t\lambda_t
=\Lambda^w.
\]
Now apply Jensen's inequality to the strictly convex function $\phi(x)=1/x$ on $(0,\infty)$:
\[
\sum_{m=1}^{M(r)} W_m(r)\phi\!\left(\bar\lambda_m^w(r)\right)
\ge
\phi\!\left(\sum_{m=1}^{M(r)} W_m(r)\bar\lambda_m^w(r)\right)
=\frac{1}{\Lambda^w}.
\]
Multiplying both sides by $\tau$ and adding $c$ gives
\[
\bar p^w(r)\ge c+\frac{\tau}{\Lambda^w}.
\]
Because $\phi$ is strictly convex, equality in Jensen's inequality holds if and only if
\[
\bar\lambda_1^w(r)=\cdots=\bar\lambda_{M(r)}^w(r)=\Lambda^w.
\]
Under the interior pricing formula this is equivalent to
\[
p_1^R(r)=\cdots=p_{M(r)}^R(r),
\]
so the constrained equilibrium price path must be constant across all hours. Since $\mathcal K(r)$ is defined as the collection of maximal consecutive constant-price blocks, this happens if and only if $M(r)=1$. Therefore every flat equilibrium path attains the lower bound of the average price $c+\tau/\Lambda^w$, while every non-flat equilibrium path yields a strictly higher weighted average price.
\end{proof}
\end{document}